\documentclass[conference]{IEEEtran}

\usepackage{etoolbox}
\makeatletter
\patchcmd{\maketitle}{\@copyrightspace}{}{}{}
\makeatother

\usepackage{rotating}
\usepackage{times}
\usepackage{microtype}
\usepackage{thmtools,thm-restate}

\usepackage{mathptmx}
\usepackage[scaled=.90]{helvet}
\usepackage{courier}
\usepackage[override]{cmtt} 

\usepackage[normalem]{ulem}
\DeclareMathAlphabet{\mathcal}{OMS}{cmsy}{m}{n}

\usepackage{mathtools}

\usepackage[usenames,dvipsnames]{xcolor}
\usepackage{soul}
\usepackage{colordvi}

\usepackage[yyyymmdd,hhmmss]{datetime}
\usepackage{array}
\usepackage[roman]{sublabel}
\usepackage[override]{cmtt}
\usepackage{cite,color,float}
\usepackage{booktabs}
\usepackage{epsfig,wrapfig,graphics,graphicx}
\usepackage{boxedminipage}
\usepackage{textcomp}
\usepackage{latexsym,fancyhdr}
\usepackage{enumerate}
\usepackage{multirow}
\usepackage{bussproofs}
\usepackage{bbm}
\usepackage{stmaryrd}

\usepackage{listings}

\usepackage{xparse} 
\usepackage{xspace}

\newif\iftr\trtrue

\usepackage{enumitem}

\setlist{noitemsep,topsep=1pt,parsep=1pt,partopsep=2pt,leftmargin=*}

\usepackage{tikz}
\usetikzlibrary{automata,positioning,chains}

\newcommand{\ignore}[1]{}

\usepackage{amsmath,amssymb,amsthm}
\usepackage{boxedminipage}
\usepackage{multirow}
\usepackage{pgfplots}
\newtheorem{claim}{Claim}
\newtheorem{definition}{Definition}

\newtheorem{theorem}{Theorem}

\usepackage{algorithm}
\usepackage[noend]{algpseudocode}
\def\BState{\State\hskip-\ALG@thistlm}


\algrenewcommand\textproc[1]{\textsf{#1}}
\algrenewcommand\algorithmicfunction{{}}

\algdef{SE}[SUBALG]{Indent}{EndIndent}{}{\algorithmicend\ }%
\algtext*{Indent}
\algtext*{EndIndent}

\usepackage{hyphenat}
\usepackage{hyperref}
\usepackage{breakurl}

\newcommand{\msf}[1]{\ensuremath{{\mathsf {#1}}}}
\renewcommand{\mtt}[1]{\ensuremath{{\mathtt {#1}}}}
\newcommand{\mcal}[1]{\ensuremath{{\mathcal {#1}}}}

\renewcommand{\P}{\mcal{P}}

\newcommand{\field}{{\ensuremath{\mathbb{F}_p}}}

\title{Brief Note: Asynchronous Verifiable Secret Sharing with Optimal Resilience and Linear Amortized Overhead}

\author{Aniket Kate \\ Purdue \and Andrew Miller \\ UIUC \and Thomas Yurek \\ UIUC}

\begin{document}

\maketitle

\begin{abstract}
  In this work we present $\msf{hbAVSS}$, the Honey Badger of Asynchronous Verifiable Secret Sharing (AVSS) protocols --- an AVSS protocol that guarantees linear amortized communication overhead even in the worst case tolerating $f < n/3$ Byzantine faults.
 The best prior work can achieve linear overhead only at a suboptimal resilience level ($t < n/4$) or by relying on optimism (falling back to quadratic overhead in case of network asynchrony or Byzantine faults).
  Our protocol therefore closes this gap, showing that linear communication overhead is possible without these compromises.
  The main idea behind our protocol is what we call the \emph{encrypt-and-disperse} paradigm: by first applying ordinary public key encryption to the secret shares, we can make use of highly efficient (but not confidentiality preserving) information dispersal primitives. We prove our protocol secure under a static computationally bounded adversary model.
\end{abstract}

\section{Introduction}
In a verifiable secret sharing (VSS) protocol, a dealer shares a secret among a set of $n$ parties, such that $t+1$ honest parties can reconstruct the secret.
VSS forms the basis for fault tolerant file storage~\cite{dispersal,pasis}, shared databases~\cite{savss}, and many other applications. It is also an essential component of secure multiparty computation (MPC) protocols, used both for generating random preprocessing elements and for accepting secret-shared inputs from untrusted clients.

The verifiability property means that if any party gets their shares, then every correct party also receives a valid share.
This is essential when VSS is used as input to MPC, since parties need to assume the shares will be available in order to make irrevocable actions (such as revealing intermediate outputs of a computation).
The challenge is that a faulty dealer may provide invalid data to some but not all of the parties.
The main idea behind nearly all VSS protocols, starting from Feldman et al.,~\cite{feldman}, is to broadcast a polynomial commitment, enabling parties to individually validate their shares.

In the case of synchronous VSS, we can simply wait to hear a confirmation from all $n$ parties, or else abort.
The asynchronous VSS case is more difficult since we must proceed after hearing from only $n-f$ of the parties, where $f$ is the number of parties that undergo a Byzantine fault.
Since crashed nodes are indistinguishable from slow nodes, it could be that $f$ of the ones we waited for are corrupted, hence only $n-2f$ correct parties received valid shares.

In order to cope with asynchrony, AVSS protocols typically distribute shares with additional redundancy, such that parties who received invalid shares can recover their shares through interaction with the others. This recovery process either results in extra communication overhead~\cite{avss}, or else relies on loosened resilience guarantees~\cite{chp}.

Improvements to AVSS have made use of concise polynomial commitments based on pairing cryptography~\cite{eavss}.
Most notably, in recent work, Basu et al.~\cite{savss} present an optimistic AVSS protocol that achieves linear communication overhead in the typical case, but in the presence of Byzantine faults or network asynchrony may fall back to quadratic overhead.
The goal of this paper is to present a protocol that provides linear guarantees even in the worst case.

{\bf Overview of our solution.}
The main idea behind our approach is a technique we call \emph{encrypt-then-disperse}, inspired by a related application in HoneyBadgerBFT~\cite{honeybadgerbft}.
The secret share encoding and polynomial commitments are as usual. However, before transmitting, the secret shares are first encrypted using public key encryption.
Next the encrypted payload is dispersed using an information dispersal routine, which can guarantee robustness and efficiency since it does not have to provide secrecy.

The use of information dispersal guarantees that every honest node receives \emph{some} data, even in the asynchronous setting. If it turns out to be invalid, then it can be used as evidence to implicate the leader.
Once the dealer is determined to be faulty, we enter a share recovery phase, which ensures every correct party receives their share. The share recovery phase can be very efficient too, since it does not need to ensure confidentiality at all, since it can only be initiated once the dealer is determined to be faulty.
A summary comparison of our results to related work is given in Table~\ref{tbl:summary}.

\begin{table}
  \centering
  \caption{Amortized communication overhead and resilience of AVSS protocols}
  \label{tbl:summary}
  \begin{tabular}{r|c|cc|}
    \multirow{2}{*}{\bfseries Protocol} &
    \multirow{2}{*}{\bfseries Resilience} & 
    \multicolumn{2}{|c|}{\bfseries Comm. Overhead} \\ \cline{3-4}
    && Typical & Worst \\ \hline
    hbAVSS (ours)      & $t<n/3$   & \multicolumn{2}{|c|}{$O(N)$}   \\
    sAVSS~\cite{savss}       & $t<n/3$   & $O(N)$ & $O(N^2)$              \\
    eAVSS~\cite{eavss}       & $t<n/3$   & \multicolumn{2}{|c|}{$O(N^2)$} \\
    AVSS~\cite{avss}         & $t<n/3$   & \multicolumn{2}{|c|}{$O(N^3)$} \\
    \msf{Sh}~\cite{chp}      & $t<n/4$   & \multicolumn{2}{|c|}{$O(N)$}   \\
\end{tabular}
\end{table}

\section{Preliminaries}

\subsection{Asynchronous Network Model}
Throughout this paper we assume the standard asynchronous network model. 
We assume a fixed set of $n$ communicating parties $\P_1,...,\P_n$, as well as a dealer $D$.
We consider a static Byzantine corruption model. The dealer and up to $f < n/3$ of the parties may be corrupted, in which case they are controlled entirely by the adversary.
The parties are connected by pairwise authenticated channels. Messages between uncorrupted parties are guaranteed eventually to be delivered, although the order and timing of delivery of messages is determined by the adversary. We assume a computationally bounded adversary that is unable to break cryptographic primitives.

\subsection{Asynchronous Verifiable Secret Sharing}
Here we give the security definition for our construction:
\begin{definition}(Asynchronous Verifiable Secret Sharing (AVSS))
  In an AVSS protocol, the dealer $D$ receives input $s \in \field$, and each party $P_i$ receives an output share $\phi(i)$ for some degree-$t$ polynomial $\phi : \field \rightarrow \field$.
\label{def:avss}
\end{definition}  

\noindent  The protocol must satisfy the following properties:
\begin{itemize}
\item \textbf{Correctness}: If the dealer $D$ is correct, then all correct parties eventually output a share $\phi(i)$ where $\phi$ is a random polynomial with $\phi(0) = s$.
\item \textbf{Secrecy}: If the dealer $D$ is correct, then the adversary learns no information about $\phi$ except for the shares of corrupted parties.
\item \textbf{Agreement}: If any correct party receives output, then there exists a unique degree-$t$ polynomial $\phi'$ such that each correct party $\P_i$ eventually outputs $\phi'(i)$.
\end{itemize}

For simplicity, this definition is written to be specific to Shamir sharing, though a more generic definition would be possible~\cite{eavss}.
Our agreement property is written to incorporate the \emph{strong commitment} property from Backes et al.~\cite{eavss}, in which the secret-shared value must be determined by the time that the first correct party outputs a share (and cannot be influenced thereafter by the adversary).

\subsection{Polynomial Commitments}
\label{sec:polycommit}
Polynomial commitments are an interface by which a committer can create a commitment to a polynomial as well as witnesses to its evaluation at different points, so to prove that evaluations are correct without revealing the full polynomial.
Polynomial commitments have been implicit in all cryptographic VSS protocols since Feldman~\cite{feldman}, but were first formalized by Kate et al.~\cite{kate2010constant}
We use the scheme from Kate et al. because it gives commitments that are additively homomorphic and constant-sized.

\begin{definition}(PolyCommit (c.f.~\cite{kate2010constant})
  Let $(\field)_\kappa$ be a family of finite fields indexed by a security parameter $\kappa$ (we'll typically omit $\kappa$ and just write $\field$). A \msf{PolyCommit} scheme for $\field$ consists of the following algorithms:
\end{definition}
\label{def:polycommit}

\begin{description}
\item [$\msf{Setup}(1^\kappa, t)$] generates system parameters $\msf{SP}$ to commit to a polynomial over $\field$ of degree bound $t$.
  $\msf{Setup}$ is run by a trusted or distributed authority. $\msf{SP}$ can also be standardized for repeated  use.
\item [$\msf{PolyCommit}(\msf{SP},\phi(\cdot))$] outputs a commitment $C$ to a polynomial
  $\phi(\cdot)$ for the system parameters \msf{SP}, and some associated decommitment information $\msf{aux}$. 
\item [$\msf{CreateWitness}(\msf{SP},\phi(\cdot),i,\msf{aux})$] outputs $\langle i,\phi(i),{w}_i
  \rangle $, where ${w}_i$ is a witness for the decommitment information for the
  evaluation $\phi(i)$ of $\phi(\cdot)$ at the index $i$.
\item [$\msf{VerifyEval}(\msf{SP},C, i,\phi(i), w_i)$] verifies that $\phi(i)$ is indeed the evaluation at the index $i$ of the polynomial committed in $C$. If so, the algorithm outputs \textit{accept}, otherwise it outputs \textit{reject}.
\end{description}
A \msf{PolyCommit} scheme must satisfy the following properties:
\begin{itemize}
    \item \textbf{Correctness}: If $C,\msf{aux} \leftarrow \msf{Commit}(\msf{SP},\phi(\cdot))$ and $w_i, \msf{aux}_i \leftarrow \msf{CreateWitness}(\msf{SP},\phi(\cdot),i,\msf{aux})$, then the correct evaluation of $\phi(i)$ is successfully verified by $\msf{VerifyEval}(\msf{SP},C, i,\phi(i), w_i, \msf{aux}_i)$.
    
    \item \textbf{Polynomial Binding}: If $C,\msf{aux} \leftarrow \msf{Commit}(\msf{SP},\phi(\cdot))$, then except with negligible probability, an adversary can not create a polynomial $\phi'(\cdot)$ such that $\msf{VerifyPoly}(\msf{SP},C,\phi(\cdot)', \msf{aux}) = 1$ if $\phi(\cdot) \neq \phi'(\cdot)$.
    
    \item \textbf{Evaluation Binding}: If $C,\msf{aux} \leftarrow \msf{Commit}(\msf{SP},\phi(\cdot))$ and $w_i, \msf{aux}_i \leftarrow \msf{CreateWitness}(\msf{SP},\phi(\cdot),i,\msf{aux})$ then except with negligible probability, an adversary can not create an evaluation $\phi(j)$, witness $w_j$, and decommitment information $\msf{aux}_j$ such that $\msf{VerifyEval}(\msf{SP},C, i,\phi(j), w_j, \msf{aux}_j) = 1$ if $i \neq j$.
    
    \item \textbf{Hiding}: Given {$C$ and $w_i$} for any $i$, an adversary either
    \begin{itemize}
        \item Can only determine $\phi(\cdot)$ or $\phi(i)$ with negligible probability given bounded computation \textit{(Computational Hiding)}
        \item Can not determine any information about $\phi(\cdot)$ or $\phi(i)$, even given unbounded computation \textit{(Unconditional Hiding)}
    \end{itemize}    
\end{itemize}

We additionally require that the commitments and witnesses be \textit{additively homomorphic}. This allows us to create new commitments and witnesses through interpolation, a property we rely on in our AVSS construction.

\begin{itemize}
    \item \textbf{Additive Homomorphism}: Given commitments $C_a$ and $C_b$ to polynomials $\phi_a(\cdot)$ and $\phi_b(\cdot)$ respectively, there should be an efficient operation to compute $C_{a+b}$, the commitment to $\phi_a(\cdot) + \phi_b(\cdot)$. Additionally, given $w_{i,a}$ and $w_{i,b}$, the witnesses for the evaluations of $\phi_a(\cdot)$ and $\phi_b(\cdot)$ at $i$ respectively, it should be similarly efficient to compute $w_{i,a+b}$. Lastly, it should also be efficient to compute $w_{i+j,a}$ from $w_{i,a}$ and $w_{j,a}$.
\end{itemize}

In this work we use PolyCommitPed from Kate et al.\cite{kate2010constant}, which provides a constant-sized commitment that achieves unconditional hiding as well as our desired homomorphic properties. We also note that PolyCommitPed achieves unconditional hiding through the use of a hiding polynomial, which we notate as $\msf{aux}$ in this work. As $\msf{aux}$ is instantiated as a polynomial over a finite field, it too realizes our desired property of additive homomorphism.

\subsection{Asynchronous Verifiable Information Dispersal}
Our protocol relies on an information dispersal protocol as defined below. Our definition is for a batch, such that $M$ messages $v_1,...,v_M$ are dispersed at once and can be individually retrieved. 
\begin{definition}
(Asynchronous Verifiable Information Dispersal (AVID))
A $(t+1,n)$ AVID scheme $\msf{AVID}$ for $M$ values is a pair of protocols $(\msf{Disperse},\msf{Retrieve})$ that satisfy the following with high probability:
\end{definition}
\begin{itemize}
\item {\bf Termination:} If the dealer $D$ is correct and initiates $\msf{Disperse}(v_1,...,v_M)$, then every correct party eventually completes $\msf{Disperse}$
\item {\bf Agreement:} If any correct party completes $\msf{Disperse}$, all correct parties eventually complete $\msf{Disperse}$.
\item {\bf Availability:} If $t+1$ correct parties have completed $\msf{Disperse}$, and some correct party initiates $\msf{Retrieve}(i)$, then the party eventually reconstructs a message $v_i'$.
\item {\bf Correctness:} After $t+1$ correct parties have completed $\msf{Disperse}$, then for each index $i \in [M]$ there is a single value $v_i$ such that if a correct party receives $v'_i$ from $\msf{Retrieve}(i)$, then $v'_i = v_i$. Furthermore if the dealer is correct, then $v_i$ is the value input by the dealer.
\end{itemize}

Hendricks et al.~\cite{avidfp} present \msf{AVID-FP}, an AVID protocol whose total communication complexity is only $O(|v|)$ in \msf{disperse} phase for a sufficiently large batch $v >> n$, i.e. it achieves only constant communication overhead.

\subsection{Reliable Broadcast}
Reliable broadcast~\cite{bracha} is a primitive that enables a dealer $D$ to broadcast a message $v$ to every party. Regardless of if the dealer is correct, if any party receives some output $v'$ then every party eventually receives $v'$.
Reliable broadcast is a special case of information dispersal, where each party simply begins \msf{Retrieve} immediately after $\msf{Disperse}$ completes.
In fact, all efficient protocols we know of, such as Cachin and Tessaro~\cite{dispersal} or Duan et al.,~\cite{duan2018beat}, are built from an AVID protocol.
We therefore skip the definition but use the $\msf{ReliableBroadcast}$ syntax in our protocol description as short hand for $\msf{Disperse}$ followed by all parties immediately beginning $\msf{Retrieve}$.

\subsection{Public Key Encryption}
We make use of a semantically secure public key encryption scheme, $(\msf{Enc},\msf{Dec})$, such that $\msf{Enc}_{\msf{PK}}(m)$ produces a ciphertext encrypted under public key $\msf{PK}$, while $\msf{Dec}_{\msf{SK}}(c)$ decrypts the message using secret key $\msf{SK}$. We assume a PKI, such that each party $\P_i$ already knows $\msf{SK}_i$.
We also assume that each public key is a function of the secret key, written $\msf{PK} = g^\msf{SK}$, which we make use of by revealing secret key during the dealer implication phase.

We note that while in our presentation we only consider a single session with a single dealer, for a practical deployment, we would want to derive per-session keys from a single long-term keypair and allow a recipient to present the session key along with a proof of its correctness, rather than reveal her secret key (and consequently need to update her key in the PKI).
\section{The \textnormal{\textsf{hbAVSS}} Protocol}

\subsection{Protocol description}
At a high level, the \msf{hbAVSS} protocol consists of the following steps:
\begin{enumerate}
\item Dealer's phase: the dealer creates Shamir sharings for $t+1$ secrets and broadcasts $t+1$ commitments to the polynomials that encode them. The dealer then encrypts each party's shares using their public encryption keys, and disperses the encrypted payloads.
\item Share validation: each party retrieves their encrypted payload, and then attempts to decrypt and validate their shares against the polynomial commitments. If sufficiently many parties successfully receive valid shares, then the shares are output.
\item Implicating a faulty dealer: if any party finds that the shares they receive are invalid or fail to decrypt, they reveal their secret key, enabling the other parties to confirm that the dealer was faulty.
\item Share recovery: once the dealer is implicated as faulty, the parties who did receive valid shares distribute them to enable the remaining parties also to reconstruct their shares.
\end{enumerate}

We now explain these steps in more detail. The protocol pseudocode is given in Algorithm~\ref{alg:hbavss}, and the narration below refers to the schematic illustration in Figure~\ref{fig:hbavss}.

\begin{figure}
  \centering
  \includegraphics[width=\columnwidth]{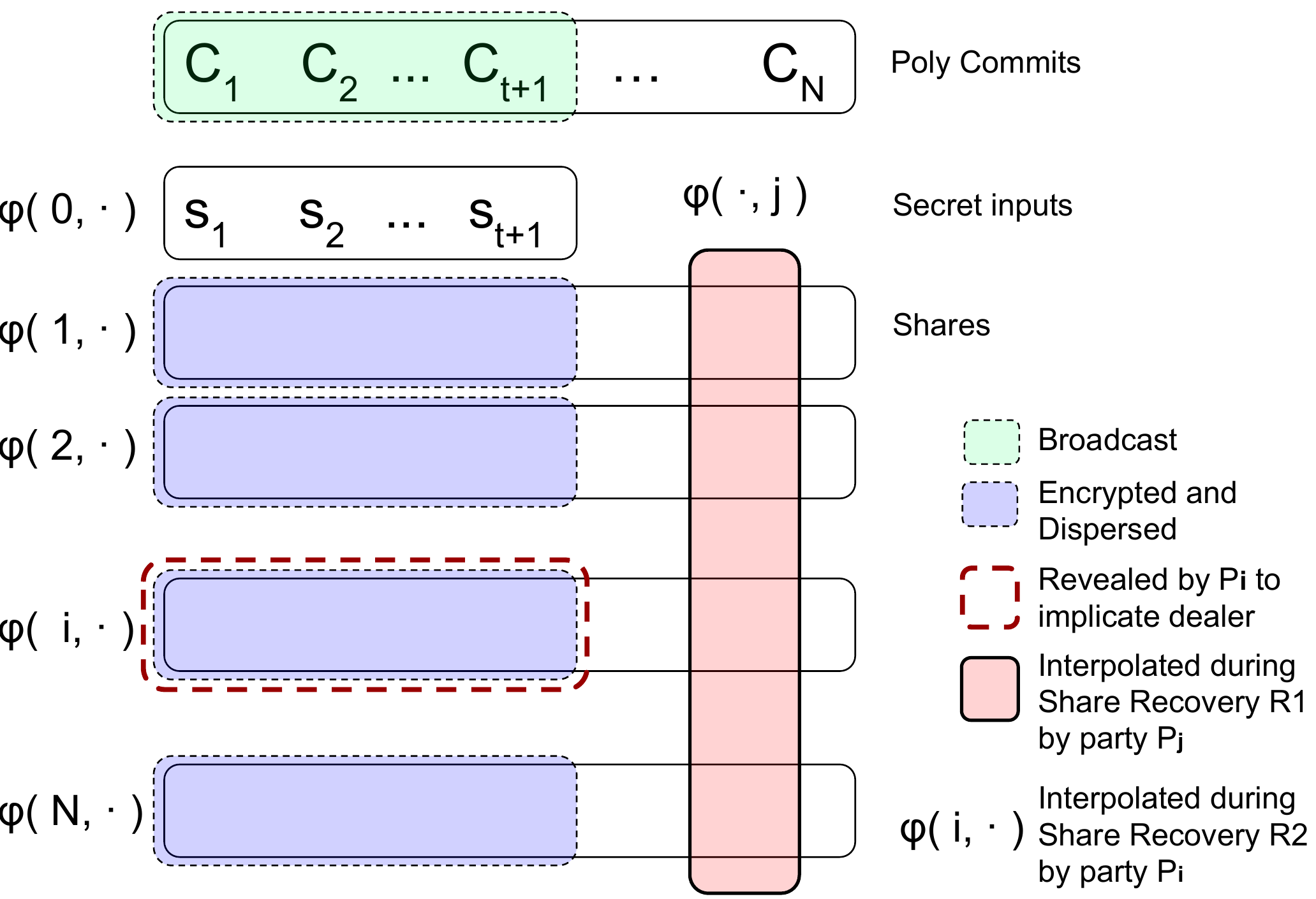}
  \label{fig:hbavss}
  \caption{Illustration of \msf{hbAVSS}. A batch of secrets $s_1,...,s_N$ are encoded as a degree-$(t,t)$ bivariate polynomial $\phi(\cdot,\cdot)$. Each party $\P_i$'s shares $\phi(i,\cdot)$ are encrypted, then dispersed. Polynomial commitments to the sharings are broadcast, enabling parties to validate their shares. If a party $\P_i$ detects an invalid share, they publish their decryption key, enabling the other parties to confirm the dealer is faulty. Invalid shares are then recovered in two steps: (R1) each party $\P_j$ receives valid shares to interpolate one column $\phi(\cdot,j)$, then (R2) each party $\P_i$ reconstructs their shares $\phi(i,\cdot)$.}
\end{figure}

\begin{algorithm*}[t]
  \caption {$\msf{hbAVSS}(D, \P_1, ..., \P_N)$ for dealer $D$ and parties $\P_1,...,\P_N$}
  \label{alg:hbavss}
  \small
  \begin{minipage}[t]{0.9\columnwidth}
    Setup:
    \begin{algorithmic}[1]
      \State Each party begins $\P_i$ with $\msf{SK}_i$ such that $\msf{PK}_i = g^{\msf{SK}_i}$
      \State The set of all $\{\msf{PK}_j\}_{j\in[N]}$ are publicly known
      \State Set up the polynomial commitment $\msf{SP} \leftarrow \msf{Setup}(t)$
    \end{algorithmic}
    \hrulefill \\
    As dealer $D$ with input $(s_1,...,s_{t+1})$:
    \begin{algorithmic}[1]
      \setcounter{ALG@line}{100}
      \Statex {\vspace{5pt} \it // Secret Share Encoding}
      \label{line:hbavss:sharing}
      \State Sample a random degree-$(t,t)$ bivariate polynomial $\phi(\cdot,\cdot)$ such that each $\phi(0,k) = s_k$ and $\phi(i,k)$ is $\P_i$'s share of $s_k$

      \Statex {\vspace{5pt} \it // Polynomial Commitment}
      \For {$k \in [t+1]$}
        \State $C_k,\msf{aux}_k \leftarrow \msf{PolyCommit}(SP, \phi(\cdot,k) )$
      \EndFor
      \State $\msf{ReliableBroadcast}^*(\{C_k\}_{k \in [t+1]})$
      \label{line:hbavss:sharingend}

      \Statex {\vspace{5pt} \it // Encrypt and Disperse}
      \For {each $\P_i$ and each $k \in [t+1]$}
        \State $w_{i,k} \leftarrow \msf{CreateWitnesss(C_k,\msf{aux}_k,i)}$
        \State $z_{i,k} \leftarrow \msf{Enc}_{\msf{PK}_i}(\phi(i,k) \| w_{i,k})$
      \EndFor
      \State $\msf{Disperse}^*(\{z_{1,k}\}_{k \in [t+1]}, ...,
         \{z_{N,k}\}_{k \in [t+1]})$
       \end{algorithmic}

       \hrulefill \\
    As receiver $\P_i$:
    \begin{algorithmic}[1]
      \setcounter{ALG@line}{200}
      \Statex {\vspace{5pt} \it // Wait for broadcasts}
      \State Wait to receive $\{C_k\}_{k \in [t+1]} \leftarrow \msf{ReliableBroadcast}^*$
      \State Wait for $\msf{Disperse}^*$ to complete
      
      \Statex {\vspace{5pt} \it // Decrypt and validate}
      \State $\{z_{i,k} \}_{k \in [t+1]} \leftarrow \msf{Retrieve}(i)$
      \label{line:hbavss:retrieve}
      \For {$k \in [t+1]$}
        \State $\phi(i,k) \| w_{i,k} \leftarrow \msf{Decrypt}_{\msf{SK}_i}(z_{i,k})$
        \If {$\msf{VerifyEval}(C_k, i, \phi(i,k), w_{i,k}) \neq 1$}
          \State {\bf sendall} $(\mtt{IMPLICATE}, \msf{SK}_i, k)$
        \EndIf
      \EndFor
      \State {\bf if} all shares were valid {\bf then} {\bf sendall} \mtt{OK}
      \label{line:hbavss:sharesvalid}
    \end{algorithmic}
  \end{minipage}\hspace{0.1\columnwidth}%
  \begin{minipage}[t]{0.9\columnwidth}
    \vspace{1pt}
    As receiver $\P_i$ (continued)
    \begin{algorithmic}[1]
      \setcounter{ALG@line}{300}
      \Statex {\vspace{5pt} \it // Bracha-style agreement}
      \State On receiving \mtt{OK} from $2t+1$ parties,
      \Indent
        \State {\bf sendall} \mtt{READY}
        \label{line:hbavss:okreceived}
      \EndIndent
      \State On receiving \mtt{READY} from $t+1$ parties,
      \Indent
        \State {\bf sendall} \mtt{READY} (if haven't yet)
        \label{line:hbavss:readyamplify}
      \EndIndent
      \State Wait to receive \mtt{READY} from $2t+1$ parties,
      \label{line:hbavss:readyok}
      \Indent
        \If {all shares were valid}
          \State {\bf output} shares $\{\phi(i,k)\}_{k\in[t+1]}$
          \label{line:hbavss:output1}
        \EndIf
      \EndIndent

      \setcounter{ALG@line}{400}
      \Statex {\vspace{5pt} \it // Handling Implication}
      \State On receiving $(\mtt{IMPLICATE}, \msf{SK}_j, k)$ from some $\P_j$,
      \Indent
        \State Discard if $\msf{PK}_j \neq g^{\msf{SK}_j}$
        \State $\{ ..., z_{j,k}, ...\} \leftarrow \msf{Retrieve}(k)$
        \vspace{6pt}
        \State $\overline{\phi(j,k)},w_{j,k} \leftarrow \msf{Decrypt}_{\msf{SK}_j}(z_{j,k})$
        \vspace{6pt}
        \State Discard if $\msf{VerifyEval}(C_k, j, \overline{\phi(j,k)}, w_{j,k}) = 1$
        \label{line:hbavss:recoveryreject}
        \State Otherwise proceed to {\it Share Recovery} below
      \EndIndent
      
      \Statex {\vspace{5pt} \it // Share Recovery}
      \setcounter{ALG@line}{500}
      \State Interpolate commitments $\{C_k\}_{k\in[N]}$ from $\{C_k\}_{k\in[t+1]}$
      \label{line:hbavss:interpcommit}
      \If {we previously received valid shares (line~\ref{line:hbavss:output1})}
        \label{line:hbavss:recover1}
        \State Interpolate witnesses $\{w_{i,k}\}_{k\in[N]}$ from $\{w_{i,k}\}_{k\in[t+1]}$
        \label{line:hbavss:interpwitness}
        \For {each $\P_j$}
          \State {\bf send} $(\mtt{R1}, \phi(i,j), w_{i,j})$ to $\P_j$
        \EndFor
        \label{line:hbavss:recover2}
      \EndIf
      \State On receiving $(\mtt{R1}, \phi(k,i), w_{k,i})$ from $t+1$ parties
      such that $\msf{VerifyEval}(C_i, k, \phi(k,i), w_{k,i}) = 1$
      \Indent
        \State Interpolate $\phi(\cdot, i)$
        \For {each $\P_j$}
          \State {\bf send} $(\mtt{R2}, \phi(j, i))$ to $\P_j$
        \EndFor
      \EndIndent
      \State On receiving $(\mtt{R2}, \phi(i, k))$ from at least $2t+1$ parties,
      \Indent
        \State Robustly interpolate $\phi(i, \cdot)$
        \State {\bf output} shares $\{\phi(i,k)\}_{k \in [t+1]}$
        \label{line:hbavss:output2}
      \EndIndent
    \end{algorithmic}
  \end{minipage}
  \vspace{7pt} \\
  $~^*$Note: To avoid clutter, the protocol is written to share a batch of exactly $t+1$ secret values. To achieve linear communication overhead, the $\msf{ReliableBroadcast}$ and $\msf{Disperse}$ instances should be shared among a batch of several simultaneously executing instances, as explained in Section~\ref{sec:performance}.  
\end{algorithm*}


\paragraph*{1) Sharing and committing}:
The protocol shares a batch of $t+1$ inputs at a time, $\{s_{1},...,s_{t+1}\}$.
The dealer creates a degree-$t$ Shamir sharing  $\phi(\cdot,k)$ for each input such that $\phi(0,k) = s_k$, and each party $P_i$'s share of $s_k$ is $\phi(i,k)$.
We visualize this as a matrix, with each party's shares forming a row as illustrated in Figure~\ref{fig:hbavss}. Later, if share recovery is needed, we make use of $\phi(\cdot,\cdot)$ as a degree-$(t,t)$ bivariate polynomial.

The dealer then uses \msf{PolyCommit} to create a commitment $C_k$ to each sharing $\phi(\cdot,k)$. The commitments are then broadcast, ensuring all the parties can validate their shares consistently.

Next, for each share $k$ and party $\P_i$, the dealer creates an encrypted payload $z_{i,k}$, consisting of the shares $\phi(i,k)$ and the polynomial evaluation witness $w_{i,k}$, encrypted under $\P_i$'s public key $\msf{PK}_i$.
The dealer then $\msf{Disperse}$s these encrypted payloads.
With the broadcast and dispersal complete, the dealer's role in the protocol is concluded --- in fact since information dispersal itself requires only one message from the dealer, the dealer's entire role is just to send messages in the first round.

\paragraph*{2) Share Verification}:
Each party $\P_i$ waits for \msf{ReliableBroadcast} and \msf{Disperse} to complete, and then retrieves just their payload $\{z_{i,k}\}_{k\in[t+1]}$.
The party then attempts to decrypt and validate its shares. If decryption is successful and all the shares are valid, then $\P_i$ signals this by sending an \mtt{OK} message to the other recipients.
The goal of the \mtt{OK} and \mtt{READY} messages (lines~\ref{line:hbavss:okreceived}-\ref{line:hbavss:output1}) is to ensure that if any party outputs a share, then enough correct parties have shares for share recovery to succeed if necessary.

\paragraph*{3) Implicating a faulty dealer}:
If any honest party $\P_i$ receives a share that either fails to decrypt or fails verification, they reveal their secret key by sending $(\mtt{IMPLICATE}, {SK}_i, k)$, which other parties can use to repeat the decryption and confirm that the dealer dispersed invalid data.

\paragraph*{4) Share Recovery}
After a dealer is implicated as faulty, the protocol enters a two-step share recovery process, following the approach of Choudhury et al.~\cite{chp}.
In the first step, parties wait for $t+1$ \mtt{R1} messages from parties that received valid shares originally.
The $\mtt{R1}$ can be checked individually by making use of the homomorphic property of polynomial commitments and witnesses (Section~\ref{sec:polycommit}).
Every correct party $\P_j$ participates in the second phase of share recovery, by reconstructing one column of the bivariate polynomial $\phi(\cdot,j)$.

The second step is the transpose, where each party reconstructs the row polynomial corresponding to its shares.
Since all correct parties send an $\mtt{R2}$ message, even if they did not originally receive valid shares, we can interpolate through ordinary robust decoding rather than using the polynomial commits.

\paragraph*{Batching}
For simplicity, we have described the protocol as sharing a batch of exactly $t+1$ values. However, to reach our desired amortized complexity goals, we need to run multiple instances in parallel in order to offset the overhead of $\msf{Disperse}$ and \msf{ReliableBroadcast}.
The idea is to run several instances of \msf{hbAVSS} such that the \msf{Disperse} and \msf{ReliableBroadcast} protocols are in lockstep, sharing their control messages (i.e., the payloads are concatenated across the several instances). 

\subsection{Security Analysis of \msf{hbAVSS}}

\begin{theorem}
  The hbAVSS protocol (Algorithm~\ref{alg:hbavss}) satifies the requirements of an AVSS protocol (with high probability) when instantiated with an additively homomorphic polynomial commitment scheme $(\msf{Setup},\msf{Commit})$, an AVID protocol $(\msf{Disperse},\msf{Retrieve})$, a reliable broadcast protocol $\msf{ReliableBroadcast}$, and a semantically secure public key encryption scheme $(\msf{Enc},\msf{Dec})$ with a pre-established PKI such that each party $\P_i$ knows their secret key $\msf{SK}_i$ and the public keys $\{ \msf{PK}_i = g^{\msf{SK}_i} \}_{i \in [N]}$ are well known.
\end{theorem}

\begin{proof}
{\bf Correctness.}
The correctness property follows easily:
If the dealer $D$ is correct, then $\msf{ReliableBroadcast}$ and $\msf{Disperse}$ complete, so each honest party receives their valid shares and outputs them through the ordinary case (line~\ref{line:hbavss:output1}).

\noindent {\bf Secrecy.}
Secrecy also follows easily.
The hiding property of the broadcasted polynomial commitments ensures that they reveal nothing about the shares.
Each party's shares are encrypted prior to dispersal, so the computationally-bounded adversary only obtains the shares that can be decrypted using corrupt parties' secret keys.
Share recovery reveals more information, but if the dealer is correct, then any attempts by the adversary to initiate share recovery will be rejected (line ~\ref{line:hbavss:recoveryreject}).

\noindent {\bf Agreement.}
It is easy to check that parties only output shares that are consistent with the broadcasted polynomial commitments. The challenge is in showing that if any correct party outputs a share, then all of them do.
In the following, assume a correct party has output a share, either through the typical path (line~\ref{line:hbavss:output1}) or through share recovery (line ~\ref{line:hbavss:output2}).
In either case, the broadcast and dispersal must have completed and the party must have received $2t+1$ \mtt{READY} messages (line~\ref{line:hbavss:readyok}).

First, notice the \mtt{READY}-amplification in line~\ref{line:hbavss:readyamplify} plays the same role as in Bracha broadcast:
\begin{claim}
  If any correct party outputs a share, then all correct parties eventually receive $2t+1$ \mtt{READY} messages (line~\ref{line:hbavss:readyok}).
  \label{claim:ready}
\end{claim}
\noindent
If any correct party receives $2t+1$ \mtt{READY} messages, then at least $t+1$ correct parties must have sent \mtt{READY} messages, which causes all correct parties to send \mtt{READY} messages.

Next, the following claim ensures that share recovery can proceed if necessary:
\begin{claim}
  If any correct party outputs a share, then at least $t+1$ correct parties receive valid shares.
  \label{claim:output1}
\end{claim}
\noindent
For \mtt{READY}-amplification to begin, some correct party must have initially sent \mtt{READY} after receiving $2t+1$ \mtt{OK} messages (line~\ref{line:hbavss:okreceived}), thus $t+1$ correct parties must have successfully received valid shares (line~\ref{line:hbavss:sharesvalid}).

Because of the availability and agreement properties of dispersal, every correct party either receives valid shares (and by then Claim~\ref{claim:ready} outputs ordinarily) or else receives an invalid share and initiates share recovery, which by Claim~\ref{claim:output1} is able to proceed.
\end{proof}

\subsection{Performance Analysis of \msf{hbAVSS}}
\label{sec:performance}
We now analyze the performance of \msf{hbAVSS}, focusing primarily on communication complexity. Recall that our goal is to achieve linear amortized communication overhead.
Since one run of this protocol results in $t+1$ secret shared values, we have a budget of $O(n^2)$ communication complexity to keep in mind.

We encounter a challenge: the up-front overhead (independent of payload size) of the \msf{Broadcast} and \msf{Disperse} primitives is either $O(n^2 \log n)$ if Merkle tree checksums are used (as in Cachin and Tessaro~\cite{dispersal}) or $O(n^3)$ if full cross checksums are used (as in Hendricks et al.~\cite{avidfp}). Hence to obtain linear overhead, we consider amortizing these costs across at least $n$ multiple runs executing in parallel, sharing the upkeep.

The broadcast payload consists of $t+1$ commitments, but each one is constant size. Since the overhead of broadcast is $O(n)$, the total communication cost is $O(n^2)$.
The total size of the \msf{Disperse} payload is $n(t+1)$ elements, but the overhead of dispersal is constant, so the total cost is $O(n^2)$. Each party \msf{Retrieve}s only a block of $t+1$ elements, thus all $n$ of them.
The Bracha-like \mtt{OK} and \mtt{READY} messages clearly contribute $O(n^2)$.

The share recovery process clearly involves $O(n^2)$ total communication, but only occurs at most once, since it enables all correct parties to reconstructing their shares at once.

The final challenge is dealing with spurious implications. In the worst case, each honest party may need to validate up to $t$ false implications, each of which requires retrieving a block from the dispersal protocol, thus $O(n^3)$ worst case cost in total. Fortunately, we amortize this in the same way, by running at least $n$ concurrent instances of the protocol. We only need to process at most one implication per party across all the instances. If spurious, further implications are ignored; if confirmed, initiate share recovery in all instances.

%


\bibliographystyle{abbrv}
\bibliography{references}

\end{document}